\documentclass[final,nomarks]{dmtcs-episciences}
\usepackage{hyperref}
\usepackage{graphicx}
\usepackage{amsmath}
\usepackage{algorithm}
\usepackage{amsthm}
\usepackage[noend]{algpseudocode}

\newtheorem{theorem}{Theorem}[section]
\newtheorem{corollary}[theorem]{Corollary}
\newtheorem{lemma}[theorem]{Lemma}

\def\O{\mathcal{O}}

\def\OT{\mathcal{\tilde O}}
\newcommand{\eexp}{\beta}
\newcommand{\STP}{\textsf{STC}}     
\newcommand{\STC}{\mathrm{STC}}     

\publicationdata{vol. 28:2}{2026}{35}{10.46298/dmtcs.16997}{2025-11-27; 2025-11-27; 2026-05-06}{2026-05-29}

\title[Approximation of Spanning Tree Congestion Using Hereditary Bisection]{Approximation of Spanning Tree Congestion Using 
Hereditary Bisection\footnote{A preliminary version of this work was presented in conference publications~\cite{kolman:24,kolman:25}. }
}

\author{Petr Kolman\thanks{Partially supported by grant 24-10306S of GA\v{C}R.}} 
\affiliation{
Faculty of Mathematics and Physics, Charles University, Prague, Czech Republic
}

\begin{document}
\maketitle              
\begin{abstract}
The \emph{Spanning Tree Congestion} (\STP) problem is the following NP-hard problem: given
a graph $G$, construct a spanning tree $T$ of $G$ minimizing its maximum edge congestion 
where the congestion of an edge $e\in T$ is the number of edges $uv$ in $G$
such that the unique path between $u$ and $v$ in $T$ passes through $e$; the optimal
value for a given graph $G$ is denoted $\STC(G)$.

It is known that \emph{every} spanning tree is an $n/2$-approximation for the
{\STP} problem. A long-standing problem is to design a better approximation
algorithm. Our contribution towards this goal is an
$\O(\Delta\cdot\log^{3/2}n)$-approximation algorithm where $\Delta$ is the
maximum degree in $G$ and $n$ the number of vertices. For graphs with a maximum
degree bounded by a polylog of the number of vertices, this is an exponential
improvement over the previous best approximation.

Our main tool for the algorithm is a new lower bound on the spanning tree
congestion which is of independent interest. Denoting by $hb(G)$ the
\emph{hereditary bisection} of $G$ which is, roughly, the maximum bisection width over all
subgraphs of $G$, we prove that for every graph $G$, $\STC(G)\geq
\Omega(hb(G)/\Delta)$. 
\end{abstract}

\section{Introduction}
The spanning tree congestion problem has been studied from various viewpoints
for more than twenty years, yet our ability to approximate it is still extremely
limited. It has been shown that every spanning tree is an
$n/2$-approximation~\cite{Otachi:20} but no $o(n)$-approximation for
\emph{general} graphs is known. There is an
$\OT(n^{1-1/(\sqrt{\log n}+1)})$-approximation\footnote{The Big-O-Tilde notation
$\mathcal{\tilde O}$ hides logarithmic factors. } algorithm for graphs with
maximum degree bounded by polylog of the number of vertices~\cite{kolman:24}. 
Chandran~et~al.~\cite{Chandranetal:18} described an algorithm that constructs in
polynomial time a spanning tree with congestion at most $\O(\sqrt{mn}\log n)$.

On the hardness side, the strongest known lower bound states that no
$c$-approximation with $c$ smaller than $6/5$ is possible unless
$P=NP$~\cite{LCH:25}. The $\Omega(n)$ gap between the best upper and lower
bounds is highly unsatisfactory.
The problem is NP-hard even on graphs of degree at most three~\cite{atalig-etal:26}.

For a detailed overview of other related results, we refer to the survey paper
by Otachi~\cite{Otachi:20}, and to the recent paper by Lampis et al.~\cite
{Lampisetal:2025} that deals with the {\STP} problem from the perspective of
parameterized complexity.

\subsection{Our Results}
Our contribution in this paper is twofold. We describe an
$\O(\Delta\cdot\log^{3/2}n)$-approximation algorithm for the spanning tree
congestion problem where $\Delta$ is the maximum degree in $G$ and $n$ the
number of vertices. For graphs with maximum degree bounded by
$\Delta=o(n/\log^{3/2}n)$, we get $o(n)$-approximation; this significantly
extends the class of graphs for which sublinear approximation is known, and
provides a partial answer to the open problem P2 from our recent
paper~\cite{kolman:24}. Moreover, for graphs with a stronger bound on the
maximum degree, the approximation ratio is even better than $o(n)$. For example,
for graphs with a maximum degree bounded by polylog of the number of vertices, the
approximation is polylogarithmic which is an exponential improvement over the
previous best approximation~\cite{kolman:24}.

For graphs excluding any fixed graph as a minor (e.g., planar graphs or bounded genus graphs),
we get a slightly better bound of $\O(\Delta\cdot\log n)$ on the approximation ratio.

Our key tool in the algorithm design is a new lower bound on $\STC(G)$ which is
our second contribution. In the recent paper~\cite{kolman:24}, we proved that
$\STC(G) \geq \frac{b(G)}{\Delta \cdot \log n}$ where $b(G)$ is the bisection of
$G$. We strengthen the bound and prove that $\STC(G) \geq
\Omega\left(\frac{hb(G)}{\Delta}\right)$ where $hb(G)$ is the \emph{hereditary
bisection} of $G$ which is, roughly, the maximum of $b(H)$ over all subgraphs $H$ of $G$.
This is a corollary of another new lower bound saying that for every subgraph
$H$ of $G$, $\STC(G) \geq \frac{\eexp(H)\cdot n'}{3\cdot \Delta }$; here
$\beta(H)$ is the expansion of $H$ and $n'$ is the number of vertices in $H$.

\subsection{Sketch of the Algorithm}\label{subsec:sketch}
The algorithm uses the standard \emph{Divide and Conquer} framework and is
conceptually very simple: partition the graph by a $\frac34$-balanced cut into
two or more connected components, solve the problem recursively for each of the
components, and arbitrarily combine the spanning trees of the components into a
spanning tree of the entire graph. The structure of the algorithm is the same as
the structure of our recent $o(n)$-approximation algorithm~\cite{kolman:24} for
graphs with maximum degree bounded by $polylog(n)$ - there is a minor difference
in the tool used in the partitioning step and in the stopping condition for the
recursion. 

It is far from obvious that the \emph{Divide and Conquer} approach works for the
spanning tree congestion problem. The difficulty is that there is no apparent
relation between $\STC(G)$ and $\STC(H)$ for a subgraph $H$ of $G$. In the
paper~\cite{kolman:24}, we proved that $\STC(G)\geq
\frac{\STC(H)}{e(H,G\setminus H)}$ where $e(H,G\setminus H)$ denotes the number
of edges between the subgraph $H$ and the rest of the graph $G$.
Note that the bound is very weak when $e(H,G\setminus H)$ is large. Also, note
that the bound is tight in the following sense: there exist graphs for which
$\STC(G)$ and $\frac{\STC(H)}{e(H,G\setminus H)}$ are equal, up to a small
multiplicative constant. For example, let $G$ be a graph obtained from a
$3$-regular expander $H$ on $n$ vertices by adding a new vertex $r$ and
connecting it by an edge to every vertex of $H$. Then $\STC(H)=\Omega(n)$ (cf.
Lemma~\ref{lem:lb1}) while $\STC(G)=O(1)$ (consider the spanning tree of $G$
consisting only of all the edges adjacent to the new vertex $r$).

The main reason for the significant improvement of the bound on the
approximation ratio of the algorithm is the new lower bound $\STC(G)\geq
\Omega\left(\frac{hb(G)}{\Delta}\right)$ that connects $\STC(G)$ and properties
of subgraphs of $G$ in a much tighter way. This connection yields a simpler
algorithm with better approximation, broader applicability and simpler analysis.

\subsection{Preliminaries}\label{subsec:Prelim}
For an undirected graph $G=(V,E)$ and a subset of vertices $S\subset V$, we
denote by $E(S, V\setminus S)$ the set of edges between $S$ and $V\setminus S$
in $G$, and by $e(S, V\setminus S)=|E(S, V\setminus S)|$ the number of these
edges. An edge $\{u,v\}\in E$ is also denoted by $uv$ for notational simplicity.
For a subset of vertices $S\subseteq V$, $G[S]$ is the subgraph induced by $S$.
By $V(G)$, we mean the vertex set of the graph $G$ and by $E(G)$ its edge set.
Given a graph $G=(V,E)$ and an edge $e\in E$, $G\setminus e$ is the graph
$(V,E\setminus \{e\})$.

Let $G=(V,E)$ be a connected graph and $T=(V,E_T)$ be a spanning tree of $G$.
For an edge $uv\in E_T$, we denote by $S_u, S_v\subset V$ the vertex sets of the
two connected components of $T\setminus uv$ containing $u$ and $v$, resp.  The
\emph{congestion $c(uv)$ of the edge $uv$  with respect to $G$ and $T$}, is the
number of edges in $G$ between $S_u$ and $S_v$.  The \emph{congestion $c(G,T)$
of the spanning tree $T$ of $G$} is defined as $\max_{e\in E_T}c(e)$, and the
\emph{spanning tree congestion $\STC(G)$ of $G$} is defined as the minimum value
of $c(G,T)$ over all spanning trees $T$ of $G$.

A \emph{bisection} of a graph with $n$ vertices is a partition of its vertices
into two sets, $S$ and $V\setminus S$, each of size at most $\lceil n/2\rceil$.
The \emph{width of a bisection} $(S,V\setminus S)$ is $e(S,V\setminus S)$. 
In approximation algorithms, the requirement that each of
the two parts in a partition of $V$ is of size at most $\lceil n/2\rceil$ is
sometimes relaxed to $2n/3$, or to some other fraction, and then we talk about
balanced cuts. In particular, a \emph{$c$-balanced cut} is a partition of the
graph vertices into two sets, each of size at most~$c\cdot n$. 
The minimum width of a $c$-balanced cut of a graph $G$ is denoted $b_{c}(G)$. 
The \emph{hereditary bisection width\/} $hb(G)$ is defined~\cite{KM:04} 
as the maximum of $b_{2/3}(H)$ over all subgraphs $H$ of $G$. 
The \emph{edge expansion\/} of $G$ is 
\begin{align}\label{def:expansion}
\eexp(G)=\min_{A\subseteq V}\frac{e(A,V\setminus A)}{\min\{|A|,|V\setminus
A|\}}\ . 
\end{align}

There are several approximation and pseudo-approximation algorithms for
bisection and balanced cuts. In our algorithm, we will employ the algorithm by
Arora, Rao and Vazirani~\cite{ARV:09}, and for graphs excluding any fixed graph
as a minor (e.g., planar graphs), a slightly stronger algorithm by Klein,
Plotkin and Rao~\cite{KleinPR:93}.

\begin{theorem}[\cite{ARV:09,KleinPR:93}]\label{thm:arv}
A $3/4$-balanced cut of cost within a ratio of $\O(\sqrt{\log n})$ of the optimum
$2/3$-balanced cut can be computed in polynomial time. For graphs excluding any fixed graph 
as a minor, even $\O(1)$ ratio is possible.
\end{theorem}

We conclude this section with two more statements that will be used later.
\begin{theorem}[Jordan~\cite{Jordan:69}]\label{thm:jordan}
Given a tree on $n$ vertices, there exists a vertex whose removal
partitions the tree into components, each with at most $n/2$ vertices.
\end{theorem}

\begin{lemma}[Kolman and Matoušek~\cite{KM:04}]\label{lem:bw-ee}
Every graph $G$ on $n$ vertices contains a subgraph on at least 
$2n/3$ vertices with edge expansion at least\footnote{Note that in the
paper~\cite{KM:04}, following earlier terminology~\cite{PachSS:96},
a $2/3$-balanced cut is called a bisection.} $b_{2/3}(G)/n$.
\end{lemma}

\bigskip
\section{New Lower Bound}
The main result of this section is captured in the following lemma and its corollary.
\begin{lemma}
\label{lem:lb1}
For every graph $G=(V,E)$ on~$n$ vertices with maximum degree~$\Delta$ and every subgraph $H$ 
of $G$ on $n'$ vertices, we have 
\begin{align}
    \STC(G) & \geq \eexp(H)\cdot \frac{n'-1}{\Delta }\ .
\end{align}
\end{lemma}

\begin{corollary}
\label{cor:lb2}
For every graph $G=(V,E)$ with maximum degree~$\Delta$, 
\begin{align}
    \STC(G) & \geq \frac{hb(G)}{3\cdot \Delta}\ .
\end{align}
\end{corollary}
Before proving the lemma and its corollary, we state a slight generalization of
Theorem~\ref{thm:jordan}; for the sake of completeness, we also provide proof of
it, though it is a straightforward extension of the standard proof of
Theorem~\ref{thm:jordan}.
\begin{lemma}\label{claim:jordan}
Given a tree $T$ on $n$ vertices with $n'\leq n$ vertices marked, there exists a vertex
(marked or unmarked) whose removal partitions the tree into components, each with at most
$n'/2$ marked vertices.
\end{lemma}
\begin{proof} 
Start with an arbitrary vertex $v_0\in T$ and set $i=0$. We proceed as follows. 
If the removal of $v_i$ partitions the tree into components such that each contains
at most $n'/2$ marked vertices, we are done. Otherwise, exactly one of the components,
say a component $C$, has strictly more than $n'/2$ marked vertices. Let $v_{i+1}$ be the neighbour
of $v_i$ that belongs to the component $C$. Note that for every $i>0$, $v_i$ is different from
all the vertices $v_0, v_1,\ldots, v_{i-1}$. As the number of vertices in the tree is
bounded, eventually, this process has to stop, and we get to a vertex with the desired properties. 
\end{proof}

\begin{proof}[of {Lemma~\ref{lem:lb1}}]
Let $T$ be the spanning tree of $G$ with the minimum congestion. By
Lemma~\ref{claim:jordan}, there exists a vertex $z\in T$ whose removal
partitions the tree $T$ into components, each with at most $n'/2$ vertices from
$H$. As the maximum degree of $G$ is $\Delta$, the number of the components is
at most $\Delta$ and, thus, at least one of them, say a component $C$, 
has at least $(n'-1)/\Delta$ vertices from $H$. 
Exploiting the definition of expansion~(\ref{def:expansion}) applied to the subgraph $H$
and its vertex subset $C\cap V(H)$ of size at least $(n'-1)/\Delta$, we get
\[
e(C,V\setminus C)\geq e(C\cap V(H),V(H)\setminus C)\geq \eexp(H)\cdot \frac{n'-1}{\Delta} \ \ .
\]
As for each edge $uv\in
E(C,V\setminus C)$, the path between $u$ and $v$ in $T$ uses 
the single edge connecting the component $C$ with the vertex $z$,
we conclude that
\[    \STC(G) \geq \eexp(H)\cdot \frac{n'-1}{\Delta }\ \ .\]
\end{proof}

\begin{proof}[of {Corollary~\ref{cor:lb2}}]
Consider a subgraph $H'$ of $G$ such that $b_{2/3}(H') = hb(G)$. By Lemma~\ref{lem:bw-ee}, 
there is a subgraph $H$ of $H'$, such that $|V(H)| \ge 2\cdot |V(H')|/3$
and $\beta(H) \ge b_{2/3}(H')/|V(H')|$. Since $H$ is a subgraph of $G$,
by Lemma~\ref{lem:lb1}, 
\[
\STC(G) \ge \beta(H) \cdot \frac{|V(H)|-1}{\Delta} \ge
\frac{b_{2/3}(H') \cdot (|V(H)|-1)}{V(H')\cdot \Delta} \ge
\frac{b_{2/3}(H')}{3 \cdot \Delta} = \frac{hb(G)}{3 \cdot \Delta} \ . 
\]
\end{proof}

Note that the bound of Corollary~\ref{cor:lb2} is tight in the following 
sense: there exist graphs for which $\STC(G)=\Theta(\frac{hb(G)}{\Delta})$. 
As in Subsection~\ref{subsec:sketch}, let $G$ be the graph obtained from a
$3$-regular expander $H$ on $n$ vertices by adding a new vertex $r$ and
connecting it by an edge to every vertex of $H$. Then $\STC(G)=4$ (the optimal
spanning tree is a star graph rooted in $r$), $hb(G)=\Theta(n)$ and $\Delta=n$.

\section{Approximation Algorithm}
Given a connected graph $G=(V,E)$, we construct a spanning tree of $G$ by the
recursive algorithm \textsc{\mbox{CongSpanTree}} called on the graph $G$. In step 3,
one of the algorithms of Theorem~\ref{thm:arv} is used: for general graphs, the
algorithm by Arora, Rao and Vazirani~\cite{ARV:09}, for graphs excluding any
fixed graph as a minor, the algorithm by Klein, Plotkin and Rao~\cite{KleinPR:93}; 
by $\alpha(n)$ we denote the respective pseudo-approximation factor.
\begin{algorithm}
\caption{\textsc{CongSpanTree}$(H)$}
\begin{algorithmic}[1]
    \If{$|V(H)|=1$}
        \State \textbf{return} $H$
    \EndIf
    \State \mbox{construct a $\frac34$-balanced cut
$(S,V(H)\setminus S)$ of $H$} 
    \State $F\gets E(S,V(H)\setminus S)$
    \For{\mbox{each connected component $C$ of $H\setminus F$}}
        \State $T_C \gets$  \textsc{CongSpanTree}$(C)$
    \EndFor
    \State {arbitrarily connect all the spanning 
	trees $T_C$ by edges from $F$ 
        to form a spanning		
	tree $T$ of $H$}
    \State \textbf{return} $T$
\end{algorithmic}
\end{algorithm}

Let $\tau$ denote the tree representing the recursive decomposition of $G$
(implicitly) constructed by the algorithm \textsc{CongSpanTree}:
The root $r$ of $\tau$ corresponds to the graph $G$, and the 
children of a non-leaf node $t\in \tau$ associated with a set $V_t$ correspond
to the connected components of $G[V_{t}] \setminus F$ where $F$
is the set of edges of the $\frac34$-balanced cut of $G[V_{t}]$ from step 4; 
by Theorem~\ref{thm:arv}, $|F|\leq \alpha(n)\cdot b_{2/3}(G[V_t])$.
We denote by $G_t=G[V_t]$ the subgraph of $G$ induced by the vertex set $V_t$,
by $T_t$ the spanning tree constructed for $G_t$ by the algorithm \textsc{CongSpanTree}.
The \emph{height} $h(t)$ of a tree node $t\in \tau$ is the number of edges on the longest 
path from $t$ to a leaf in its subtree (i.e., to a leaf that is a descendant of $t$). 

\begin{lemma}\label{lem:recurse}
Let $t\in \tau$ be a node of the decomposition tree and $t_1,\ldots,t_k$ its children. 
Then
\begin{align}
    c(G_t,T_t)\leq \max_{i}c(G_{t_i},T_{t_i}) + \alpha(n)\cdot b_{2/3}(G_t) \ .
\end{align}    
\end{lemma}
\begin{proof}
Let $F$ be the set of edges of the $\frac34$-balanced cut of $G_t$ from step 4.
We will show that for every edge $e\in E(T_t)$, its congestion $c(e)$ with
respect to $G_t$ and $T_t$ is at most $\max_{i}c(G_{t_i},T_{t_i}) + |F|$; as
$|F|\leq \alpha(n)\cdot b_{2/3}(G[V_t])$, this will prove the lemma. Recall that
$E(T_t)\subseteq \bigcup_{i=1}^k E(T_{t_i}) \cup F$, as the spanning tree $T_t$
is constructed (step 7) from the spanning trees $ T_{t_1},\ldots, T_{t_k}$ and
the set~$F$.

Consider first an edge $e\in E(T_t)$ that belongs to a tree $T_{t_i}$, for some
$i$. The only edges from $E(G)$ that may contribute to the congestion $c(e)$ of
$e$ with respect to $G_t$ and $T_t$ are the edges in $E(G_{t_i})\cup F$; the
contribution of the edges in $E(G_{t_i})$ is at most $c(G_{t_i},T_{t_i})$, the
contribution of the edges in $F$ is at most $|F|$. Thus, the congestion $c(e)$
of the edge $e$ with respect to $G_t$ and $T_t$ is at most
$c(G_{t_i},T_{t_i})+|F|$.

Consider now an edge $e\in F\cap E(T_t)$. As the only edges from $E(G)$ that may contribute
to the congestion $c(e)$ of $e$ with respect to $G_t$ and $T_t$ are the edges in $F$, 
its congestion is at most $|F|$. 

Thus, for every edge $e\in E(T_t)$, its congestion with respect to $G_t$ and
$T_t$ is at most $\max_{i}c(G_{t_i},T_{t_i}) + |F|$, and the proof of the lemma
is completed.
\end{proof}

\begin{lemma}\label{cor:stc}
Let $T=\textsc{CongSpanTree}(G)$. Then
\begin{align}
    c(G,T)\leq \O(\alpha(n)\cdot \log n)\cdot hb(G) \ .
\end{align}
\end{lemma}
\begin{proof}
For technical reasons, we extend the notion of the spanning tree congestion also to the 
trivial graph $H=(\{v\},\emptyset)$ consisting of a single vertex and no edge (and having 
a single spanning tree $T_H=H$) by defining $c(H,T_H)=0$.

By induction on the height of vertices in the decomposition tree $\tau$, 
we prove the following auxiliary claim: for every $t\in \tau$, 
\begin{align}\label{eq:recur}
    c(G_t,T_t) \leq  h(t)\cdot \alpha(n)\cdot hb(G) \ .
\end{align}
Consider first a node $t\in \tau$ of height zero, that is, a node $t$ that is a leaf. 
Then both sides of~(\ref{eq:recur}) are zero and the inequality holds.

Consider now a node $t\in \tau$ such that for all his children
the inequality~(\ref{eq:recur}) holds. Let $t'$ be the child of the node $t$ for which
$c(G_{t'},T_{t'})$ is the largest among the children of $t$. Then, 
as $b_{2/3}(G_t)\leq hb(G)$ by the definition of $hb$, by Lemma~\ref{lem:recurse} we get
\[
    c(G_t,T_t) \leq c(G_{t'},T_{t'}) +  \alpha(n) \cdot hb(G) \ .
\]
By the inductive assumption applied on the node $t'$, 
\[
    c(G_{t'},T_{t'}) \leq  h(t')\cdot \alpha(n)\cdot hb(G) \ .
\]
Because $h(t')+1\leq h(t)$, the proof of the auxiliary claim is completed.

Observing that the height of the root of the decomposition tree $\tau$ is at
most $\O(\log n)$, as all cuts used by the algorithm are balanced, the proof is
completed.
\end{proof}

\begin{theorem}\label{thm:algorithm}
Given a graph $G$ with maximum degree $\Delta$, the algorithm
\textsc{CongSpanTree} constructs an $\O(\Delta\cdot \log^{3/2} n)$-approximation of
the minimum congestion spanning tree; for graphs excluding any fixed graph as a minor,
the approximation is $\O(\Delta\cdot \log n)$.
\end{theorem}

\begin{proof}
By Corollary~\ref{cor:lb2}, for every graph $G$, $\Omega(hb(G)/\Delta)$ is a
lower bound on $\STC(G)$. By Lemma~\ref{cor:stc}, the algorithm
$\textsc{CongSpanTree}(G)$ constructs a spanning tree $T$ of congestion at most
$\O(\alpha(n)\cdot \log n)\cdot hb(G)$. Combining these two results yields the
theorem: $c(G,T)\leq \O(\alpha(n) \cdot \log n)\cdot hb(G) \leq \O(\Delta\cdot\alpha(n)
\cdot \log n )\cdot \STC(G)$. Plugging in the bounds on $\alpha(n)$
from Theorem~\ref{thm:arv} yields the theorem.
\end{proof}

\section{Conclusion and Open Problems}
We have designed an $o(n)$-approximation algorithm for the spanning tree
congestion problem for graphs with maximum degree bounded by
$o(n/\log^{3/2}n)$.  An inevitable question is whether it is possible to
eliminate or weaken the dependence of the approximation ratio of the algorithm on the
largest degree and obtain an $o(n)$-approximation algorithm for the {\STP}
problem for all graphs.
A related challenge is to establish a stronger lower bound on the
approximability of the $\STP$ problem for graphs with large degrees. 

A different kind of problem is to find other applications of the hereditary
bisection in the context of approximation algorithms.

To state the last problem, we first define, following Law and Ostrovskii~\cite{LO:10}, 
\begin{align*}
    f(G,c)=\min\{e(S,V\setminus S)\ | \ S\subset V, |S|=c, \mbox{ $G[S]$ is connected}\} \ ,
\end{align*}
where $G=(V,E)$ is a given graph and $c$ is a given integer.
They proved the following lower bound:
\begin{align}
      \STC(G) & \geq \min\left \{f(G,c)  \ \middle| \ \left \lceil\frac{n-1}{\Delta}\right \rceil \leq c \leq \frac{n}{2} \right \}\ \mbox{. } \label{lb:connected-sized-cuts}
\end{align}
Note that $f(G,c)$ is the size of a cut $E(S,V\setminus S)$ in $G$ satisfying three properties:
\begin{itemize}
\item[i)] the subgraph of $G$ induced by $S$ is \emph{connected}, 
\item[ii)] the subset $S$ has a \emph{prescribed size}, and
\item[iii)] the number of edges $e(S, V\setminus S)$ is \emph{the smallest}
among all subsets satisfying the properties i) and~ii).
\end{itemize}
We call the task of finding such a cut the {\em minimum connected $c$-cut} problem.
If only any two of these three properties are required, 
then the problem of finding such a cut is solvable, or
at least well approximable, in polynomial time:
\begin{itemize}
\item For the minimum $c$-cut, not necessarily connected, Feige and Krauthgamer~\cite{FK:02}
give a polylogarithmic approximation. 
\item For the minimum connected cut without the size constraint, a standard minimum 
cut algorithm solves the problem (cf.~\cite{KT:26}).
\item A subset $S\subseteq V$ of size $c$ inducing a connected graph, not necessarily minimizing
$e(S,V\setminus S)$, can be constructed greedily. 
\end{itemize}
However, we are not aware of any non-trivial approximation if all
three requirements have to be at least approximately satisfied.
Thus, another open problem is to design a non-trivial approximation algorithm for the minimum 
connected $c$-cut problem, and then to use it for a better approximation of the {\STP}
problem.

\acknowledgements
The authors would like to thank the anonymous referees for their constructive comments.

\bibliographystyle{plainurl}
\bibliography{stc-bibliography-doi}
\end{document}